\documentclass[12pt,letterpaper]{article}
\usepackage{amsmath,amsfonts,amsthm,amssymb,tabls}

\theoremstyle{plain}

\numberwithin{equation}{section}
\newtheorem{thm}{Theorem}[section]
\newtheorem{lem}[thm]{Lemma}
\newtheorem{cor}[thm]{Corollary}

\newenvironment{exam}[1]%
{\begin{flushleft}\textbf{Example #1}.\enspace}%
{\end{flushleft}}

\newcommand{\complex}{{\mathbb C}}
\newcommand{\positive}{{\mathbb N}}
\newcommand{\real}{{\mathbb R}}
\newcommand{\ascript}{{\mathcal A}}
\newcommand{\bscript}{{\mathcal B}}
\newcommand{\cscript}{{\mathcal C}}
\newcommand{\pscript}{{\mathcal P}}
\newcommand{\qscript}{{\mathcal Q}}
\newcommand{\sscript}{{\mathcal S}}
\newcommand{\vscript}{{\mathcal V}}
\newcommand{\rmtr}{\mathrm{tr}}
\newcommand{\rmcyl}{\mathrm{cyl}}
\newcommand{\rmrank}{\mathrm{rank}}

\newcommand{\muhat}{\widehat{\mu}}

\newcommand{\ab}[1]{\left|#1\right|}
\newcommand{\doubleab}[1]{\left\|#1\right\|}
\newcommand{\brac}[1]{\left\{#1\right\}}
\newcommand{\paren}[1]{\left(#1\right)}
\newcommand{\sqbrac}[1]{\left[#1\right]}
\newcommand{\elbows}[1]{{\left\langle#1\right\rangle}}
\newcommand{\ket}[1]{{\left|#1\right>}}
\newcommand{\bra}[1]{{\left<#1\right|}}

\begin{document}

\title{MODELS FOR\\DISCRETE QUANTUM GRAVITY
}
\author{S. Gudder\\ Department of Mathematics\\
University of Denver\\ Denver, Colorado 80208, U.S.A.\\
sgudder@du.edu
}
\date{}
\maketitle

\begin{abstract}
We first discuss a framework for discrete quantum processes (DQP). It is shown that the set of $q$-probability operators is convex and its set of extreme elements is found. The property of consistency for a DQP is studied and the quadratic algebra of suitable sets is introduced. A classical sequential growth process is ``quantized'' to obtain a model for discrete quantum gravity called a quantum sequential growth process (QSGP). Two methods for constructing concrete examples of QSGP are provided.
\end{abstract}

\section{Introduction}  
In a previous article, the author introduced a general framework for a discrete quantum gravity \cite{gud11}. However, we did not include any concrete examples or models for this framework. In particular, we did not consider the problem of whether nontrivial models for a discrete quantum gravity actually exist. In this paper we provide a method for constructing an infinite number of such models. We first make a slight modification of our definition of a discrete quantum process (DQP) $\rho _n$, $n=1,2,\ldots\,$. Instead of requiring that $\rho _n$ be a state on a Hilbert space $H_n$, we require that $\rho _n$ be a $q$-probability operator on $H_n$. This latter condition seems more appropriate from a probabilistic viewpoint and instead of requiring $\rmtr (\rho _n)=1$, this condition normalizes the corresponding quantum measure. By superimposing a concrete DQP on a classical sequential growth process we obtain a model for discrete quantum gravity that we call a quantum sequential growth process.

Section~2 considers the DQP formalism. We show that the set of $q$-probability operators is a convex set and find its set of extreme elements. We discuss the property of consistency for a DQP and introduce the so-called quadratic algebra of suitable sets. The suitable sets are those on which well-defined quantum measures (or quantum probabilities) exist.

Section~3 reviews the concept of a classical sequential growth process (CSGP) \cite{blms87, hen09, rs00, sor03, sur11, vr06}. The important notions of paths and cylinder sets are discussed. In Section~4 we show how to ``quantize'' a CSGP to obtain a quantum sequential growth process (QSGP). Some results concerning the consistency of a DQP are given. Finally, Section~5 provides two methods for constructing examples of QSGP.

\section{Discrete Quantum Processes} 
Let $(\Omega ,\ascript ,\nu )$ be a probability space and let
\begin{equation*}
H=L_2(\Omega ,\ascript ,\nu )=\brac{f\colon\Omega\to\complex ,\int\ab{f}^2d\nu<\infty}
\end{equation*}
be the corresponding Hilbert space. Let $\ascript _1\subseteq\ascript _2\subseteq\cdots\subseteq\ascript$ be an increasing sequence of sub $\sigma$-algebras of $\ascript$ that generate $\ascript$ and let
$\nu _n=\nu\mid\ascript _n$ be the restriction of $\nu$ to $\ascript _n$, $n=1,2,\ldots\,$. Then
$H_n=L_2(\Omega ,\ascript _n,\nu _n)$ forms an increasing sequence of closed subspaces of $H$ called a
\textit{filtration} of $H$. A bounded operator $T$ on $H_n$ will also be considered as a bounded operator on $H$ by defining $Tf=0$ for all $f\in H_n^\perp$. We denote the characteristic function $\chi _\Omega$ of $\Omega$ by $1$. Of course, $\doubleab{1}=1$ and $\elbows{1,f}=\int fd\nu$ for every $f\in H$. A $q$-\textit{probability operator} is a bounded positive operator $\rho$ on $H$ that satisfies $\elbows{\rho 1,1}=1$. Denote the set of $q$-probability operators on $H$ and $H_n$ by $\qscript (H)$ and $\qscript (H_n)$, respectively. Since $1\in H_n$, if
$\rho\in\qscript (H_n)$ by our previous convention, $\rho\in\qscript (H)$. Notice that a positive operator
$\rho\in\qscript (H)$ if and only if $\doubleab{\rho ^{1/2}1}=1$ where $\rho ^{1/2}$ is the unique positive square root of $\rho$.

A $\rmrank~1$ element of $\qscript (H)$ is called a \textit{pure} $q$-\textit{probability operator}. Thus $\rho\in\qscript (H)$ is pure if and only if $\rho$ has the form $\rho =\ket{\psi}\bra{\psi}$ for some $\psi\in H$ satisfying
\begin{equation*}
\ab{\elbows{1,\psi}}=\ab{\int\psi d\nu}=1
\end{equation*}
We then call $\psi$ a $q$-\textit{probability vector} and we denote the set of $q$-probability vectors by
$\vscript (H)$ and the set of pure $q$-probability operators by $\qscript _p(H)$. Notice that if $\psi\in\vscript (H)$, then
$\doubleab{\psi}\ge 1$ and $\doubleab{\psi}=1$ if and only if $\psi =\alpha 1$ for some $\alpha\in\complex$ with
$\ab{c}=1$. Two operators $\rho _1,\rho _2\in\qscript (H)$ are \textit{orthogonal} if $\rho _1\rho _2=0$.

\begin{thm}       
\label{thm21}
{\rm (i)}\enspace $\qscript (A)$ is a convex set and $\qscript _p(H)$ is its set of extreme elements.
{\rm (ii)}\enspace $\rho\in\qscript (H)$ is of trace class if and only if there exists a sequence of mutually orthogonal
$\rho _i\in\qscript _p(H)$ and $\alpha _i>0$ with $\sum\alpha _i=1$ such that $\rho =\sum\alpha _i\rho _i$ in the strong operator topology. The $\rho _i$ are unique if and only if the corresponding $\alpha _i$ are distinct.
\end{thm}
\begin{proof}
(i)\enspace If $0<\lambda <1$ and $\rho _1,\rho _2\in\qscript (H)$, then $\rho =\lambda\rho _1+(1-\lambda )\rho _2$ is a positive operator and
\begin{equation*}
\elbows{\rho 1,1}=\elbows{\paren{\lambda\rho +(1-\lambda}\rho _2)1,1}
  =\lambda\elbows{\rho _11,1}+(1-\lambda )\elbows{\rho _21,1}=1
\end{equation*}
Hence, $\rho\in\qscript (H)$ so $\qscript (H)$ is a convex set. Suppose $\rho\in\qscript _p(H)$ and
$\rho =\lambda\rho _1+(1-\lambda )\rho _2$ where $0<\lambda <1$ and $\rho _1,\rho _2\in\qscript (H)$. If
$\rho _1\ne\rho _2$, then $\rmrank (\rho )\ne 1$ which is a contradiction. Hence, $\rho _1=\rho _2$ so $\rho$ is an extreme element of $\qscript (H)$. Conversely, suppose $\rho\in\qscript (H)$ is an extreme element. If the cardinality of  the spectrum $\ab{\sigma (\rho )}>1$, then by the spectral theorem $\rho =\rho _1+\rho _2$ where
$\rho _1,\rho _2\ne 0$ are positive and $\rho _1\ne\alpha\rho _2$ for $\alpha\in\complex$. If
$\rho _11,\rho _21\ne 0$, then $\elbows{\rho _11,1},\elbows{\rho _21,1}\ne 0$ and we can write
\begin{equation*}
\rho =\elbows{\rho _11,1}\frac{\rho _1}{\elbows{\rho _11,1}}+\elbows{\rho _21,1}\frac{\rho _2}{\elbows{\rho _21,1}}
\end{equation*}
Now $\elbows{\rho _11,1}^{-1}\rho _1,\elbows{\rho _21,1}^{-1}\rho _2\in\qscript (H)$ and
\begin{equation*}
\elbows{\rho _11,1}+\elbows{\rho _21,1}=\elbows{\rho 1,1}=1
\end{equation*}
which is a contradiction. Hence, $\rho _11=0$ or $\rho _21=0$. Without loss of generality suppose that $\rho _21=0$. We can now write
\begin{equation*}
\rho =\tfrac{1}{2}\rho _1+\tfrac{1}{2}(\rho _1+2\rho _2)
\end{equation*}
Now $\rho _11\ne 0$, $(\rho _1+2\rho _2)1\ne 0$ and as before we get a contradiction. We conclude that
$\ab{\sigma (\rho )}=1$. Hence, $\rho =\alpha P$ where $P$ is a projection and $\alpha >0$. If $\rmrank (P)>1$, then
$P=P_1+P_2$ where $P_1$ and $P_2$ are orthogonal nonzero projections so $\rho =\alpha P_1+\alpha P_2$. Proceeding as before we obtain a contradiction. Hence, $\rmrank (P)=1$ so $\rho =\alpha P$ is pure.
(ii)\enspace This follows from the spectral theorem.
\end{proof}

Let $\brac{H_n\colon n=1,2,\ldots}$ be a filtration of $H$ and let $\rho _n\in\qscript (H_n)$, $n=1,2,\ldots\,$. The
$n$-\textit{decoherence functional} $D_n\colon\ascript _n\times\ascript _n\to\complex$ defined by
\begin{equation*}
D_n(A,B)=\elbows{\rho _n\chi _B,\chi _A}
\end{equation*}
gives a measure of the interference between $A$ and $B$ when the system is described by $\rho _n$. It is clear that $D_n(\Omega _n,\Omega _n)=1$, $D_n(A,B)=\overline{D_n(B,A)}$ and $A\mapsto D_n(A,B)$ is a complex measure for all $B\in\ascript _n$. It is also well-known that if $A_1,\ldots ,A_r\in\ascript _n$ then the matrix with entries
$D_n(A_j,A_k)$ is positive semidefinite. We define the map $\mu _n\colon\ascript _n\to\real ^+$ by
\begin{equation*}
\mu _n(A)=D_n(A,A)=\elbows{\rho _n\chi _A,\chi _A}
\end{equation*}
Notice that $\mu _n(\Omega _n)=1$. Although $\mu _n$ is not additive, it does satisfy the \textit{grade}-2
\textit{additivity condition}: if $A,B,C\in\ascript _n$ are mutually disjoint, then
\begin{align}         
\label{eq21}
\mu _n(A\cup B\cup C)&=\mu _n(A\cup B)+\mu _n(A\cup C)+\mu _n(B\cup C)\notag\\
  &\qquad-\mu _n(A)-\mu _n(B)-\mu _n(C)
\end{align}
 We say that $\rho _{n+1}$ is \textit{consistent} with $\rho _n$ if $D_{n+1}(A,B)=D_n(A,B)$ for all $A,B\in\ascript _n$. We call the sequence $\rho _n$, $n=1,2,\ldots$, \textit{consistent} if $\rho _{n+1}$ is consistent with $\rho _n$ for
$n=1,2,\ldots\,$. Of course, if the sequence $\rho _n$, $n=1,2,\ldots$, is consistent, then $\mu _{n+1}(A)=\mu _n(A)$
$\forall A\in\ascript _n$, $n=1,2,\ldots\,$. A \textit{discrete quantum process} (DQP) is a consistent sequence
$\rho _n\in\qscript (H_n)$ for a filtration $H_n$, $n=1,2,\ldots\,$. A DQP $\rho _n$ is \textit{pure} if
$\rho _n\in\qscript _p(H_n)$, $n=1,2,\ldots\,$.

If $\rho _n$ is a DQP, then the corresponding maps $\mu _n\colon\ascript _n\to\real ^+$ have the form
\begin{equation*}
\mu _n(A)=\elbows{\rho _n\chi _A,\chi _A}=\doubleab{\rho _n^{1/2}\chi _A}^2
\end{equation*}
Now $A\to\rho _n^{1/2}\chi _A$ is a vector-valued measure on $\ascript _n$. We conclude that $\mu _n$ is the squared norm of a vector-valued measure. In particular, if $\rho _n=\ket{\psi _n}\bra{\psi _n}$ is a pure DQP, then
$\mu _n(A)=\ab{\elbows{\psi _n,\chi _A}}^2$ so $\mu _n$ is the squared modulus of the complex-valued measure
$A\mapsto\elbows{\psi _n,\chi _A}$.

For a DQP $\rho _n\!\in\!\qscript (H_n)$, we say that a set $A\!\in\!\ascript$ is \textit{suitable} if
$\lim\elbows{\rho _j\chi _A,\chi _A}$ exists and is finite and in this case we define $\mu (A)$ to be the limit. We denote the set of suitable sets by $\sscript (\rho _n)$. If $A\in\ascript _n$ then
\begin{equation*}
\lim\elbows{\rho _j\chi _A,\chi _A}=\elbows{\rho _n\chi _A,\chi _A}
\end{equation*}
so $A\in\sscript (\rho _n)$ and $\mu (A)=\mu _n(A)$. This shows that the algebra
$\ascript _0=\cup\ascript _n\subseteq\sscript (\rho _n)$. In particular, $\Omega\in\sscript (\rho _n)$ and
$\mu (\Omega )=1$. In general, $\sscript (\rho _n)\ne\ascript$ and $\mu$ may not have a well-behaved extension from
$\ascript _0$ to all of $\ascript$ \cite{dgt08, sor11}. A subset $\bscript$ of $\ascript$ is a \textit{quadratic algebra} if
$\emptyset,\Omega\in\bscript$ and whenever $A,B,C\in\bscript$ are mutually disjoint with
$A\cup B,A\cup C,B\cup C\in\bscript$, we have $A\cup B\cup C\in\bscript$. For a quadratic algebra $\bscript$, a
$q$-\textit{measure} is a map $\mu _0\colon\bscript\to\real ^+$ that satisfies the grade-2 additivity condition
\eqref{eq21}. Of course, an algebra of sets is a quadratic algebra and we conclude that
$\mu _n\colon\ascript _n\to\real ^+$ is a $q$-measure. It is not hard to show that $\sscript (\rho _n)$ is a quadratic algebra and $\mu\colon\sscript (\rho _n)\to\real ^+$ is a $q$-measure on $\sscript (\rho _n)$ \cite{gud11}.

\section{Classical Sequential Growth Processes} 
A \textit{partially ordered set} (\textit{poset}) is a set $x$ together with an irreflexive, transitive relation $<$ on $x$. In this work we only consider unlabeled posets and isomorphic posets are considered to be identical. Let $\pscript _n$ be the collection of all posets with cardinality $n$, $n=1,2,\ldots\,$. If $x\in\pscript _n$, $y\in\pscript _{n+1}$, then $x$ \textit{produces} $y$ if $y$ is obtained from $x$ by adjoining a single new element to $x$ that is maximal in $y$. We also say that $x$ is a \textit{producer} of $y$ and $y$ is an \textit{offspring} of $x$. If $x$ produces $y$ we write
$x\to y$. We denote the set of offspring of $x$ by $x\to$ and for $A\subseteq\pscript _n$ we use the notation
\begin{equation*}
A\to\ =\brac{y\in\pscript _{n+1}\colon x\to y, x\in A}
\end{equation*}
The transitive closure of $\to$ makes the set of all finite posets $\pscript =\cup\pscript _n$ into a poset.

A \textit{path} in $\pscript$ is a string (sequence) $x_1,x_2,\ldots$ where $x_i\in\pscript _i$ and $x_i\to x_{i+1}$, $i=1,2,\ldots\,$. An $n$-\textit{path} in $\pscript$ is a finite string $x_1x_2\cdots x_n$ where again $x_i\in\pscript _i$ and $x_i\to x_{i+1}$. We denote the set of paths by $\Omega$ and the set of $n$-paths by $\Omega _n$. The set of paths whose initial $n$-path is $\omega _0\in\Omega _n$ is denoted by $\omega _0\Rightarrow$. Thus, if
$\omega _0=x_1x_2\cdots x_n$ then
\begin{equation*}
\omega _0\Rightarrow =\brac{\omega\in\Omega\colon\omega =x_1,x_2\cdots x_ny_{n+1}y_{n+2}\cdots}
\end{equation*}
If $x$ produces $y$ in $r$ isomorphic ways, we say that the \textit{multiplicity} of $x\to y$ is $r$ and write
$m(x\to y)=r$. For example, in Figure~1, $m(x\to y)=3$. (To be precise, these different isomorphic ways require a labeling of the posets and this is the only place that labeling needs to be mentioned.)
\vskip 7pc

\setlength{\unitlength}{8pt}
\begin{picture}(25,12)
\put(22,8){\circle{7}}   
\put(20.5,8){\circle*{.35}}
\put(22,8){\circle*{.35}}
\put(23.5,8){\circle*{.35}}
\put(22,4){\makebox{$x$}}
\put(22,10.5){\vector(0,1){4}}
\put(19.5,8){\vector(-1,1){7.5}}
\put(24.5,8){\vector(1,1){7.5}}
\put(10,17){\circle{7}}    
\put(8.5,16.2){\circle*{.35}}
\put(10,16.2){\circle*{.35}}
\put(11.5,16.2){\circle*{.35}}
\put(8.5,16.2){\line(1,2){.8}}
\put(10,16.2){\line(-1,2){.8}}
\put(9.23,17.85){\circle*{.35}}
\put(22,17){\circle{7}}        
\put(20.5,16.2){\circle*{.35}}
\put(22,16.2){\circle*{.35}}
\put(23.5,16.2){\circle*{.35}}
\put(23.5,16.2){\line(-1,2){.8}}
\put(22,16.2){\line(1,2){.8}}
\put(22.75,17.8){\circle*{.35}}
\put(34,17){\circle{7}}        
\put(32.5,16.2){\circle*{.35}}
\put(34,16.2){\circle*{.35}}
\put(35.5,16.2){\circle*{.35}}
\put(32.5,16.2){\line(1,1){1.5}}
\put(35.5,16.2){\line(-1,1){1.5}}
\put(34,17.7){\circle*{.35}}
\centerline{\textbf{\hskip -2pc Figure 1}}
\end{picture}
\vskip 2pc

If $x\in\pscript$ and $a,b\in x$ we say that $a$ is an \textit{ancestor} of $b$ and $b$ is a \textit{successor} of $a$ if
$a<b$. We say that $a$ is a \textit{parent} of $b$ and $b$ is a \textit{child} of $a$ if $a<b$ and there is no $c\in x$ such that $a<c<b$. Let $c=(c_0,c_1,\ldots )$ be a sequence of nonnegative numbers called \textit{coupling constants}
\cite{rs00, vr06}. For $r,s\in\positive$ with $r\le s$, we define
\begin{equation*}
\lambda _c(s,r)=\sum _{k=r}^s\binom{s-r}{k-r}c_k=\sum _{k=0}^{s-r}\binom{s-r}{k}c_{r+k}
\end{equation*}
For $x\in\pscript _n$ $y\in\pscript _{n+1}$ with $x\to y$ we define the \textit{transition probability}
\begin{equation*}
p _c(x\to y)=m(x\to y)\frac{\lambda _c(\alpha ,\pi )}{\lambda _c(n,0)}
\end{equation*}
where $\alpha$ is the number of ancestors and $\pi$ the number of parents of the adjoined maximal element in $y$ that produces $y$ from $x$. It is shown in \cite{rs00, vr06} that $p_c(x\to y)$ is a probability distribution in that it satisfies the Markov-sum rule
\begin{equation*}
\sum\brac{p_c(x\to y)\colon y\in\pscript _{n+1},x\to y}=1
\end{equation*}

In discrete quantum gravity, the elements of $\pscript$ are thought of as causal sets and $a<b$ is interpreted as $b$ being in the causal future of $a$. The distribution $y\mapsto p_c(x\to y)$ is essentially the most general that is consistent with principles of causality and covariance \cite{rs00, vr06}. It is hoped that other theoretical principles or experimental data will determine the coupling constants. One suggestion is to take $c_k=1/k!$ \cite{sor03, sor11}. The case $c_k=c^k$ for some $c>0$ has been previously studied and is called a \textit{percolation dynamics}
\cite{rs00, sor03, sur11}.

We call an element $x\in\pscript$ a \textit{site} and we sometimes call an $n$-path an $n$-\textit{universe} and a path a \textit{universe} The set $\pscript$ together with the set of transition probabilities $p_c(x\to y)$ forms a
\textit{classical sequential growth process} (CSGP) which we denote by $(\pscript ,p_c)$
\cite{hen09, rs00, sor03, sur11, vr06}. It is clear that $(\pscript ,p_c)$ is a Markov chain and as usual we define the probability of an $n$-path $\omega =y_1y_2\cdots y_n$ by
\begin{equation*}
p_c^n(\omega )=p_c(y_1\to y_2)p_c(y_2\to y_3)\cdots p_c(y_{n-1}\to y_n)
\end{equation*}
Denoting the power set of $\Omega _n$ by $2^{\Omega _n}$, $(\Omega _n,2^{\Omega _n},p_c^n)$ becomes a probability space where
\begin{equation*}
p_c^n(A)=\sum\brac{p_c^n(\omega )\colon\omega\in A}
\end{equation*}
for all $A\in2^{\Omega _n}$. The probability of a site $x\in\pscript _n$ is
\begin{equation*}
p_c^n(x)=\sum\brac{p_x^n(\omega )\colon\omega\in\Omega _n,\omega\hbox{ ends at }x}
\end{equation*}
Of course, $x\mapsto p_c^n(x)$ is a probability measure on $\pscript _n$ and we have
\begin{equation*}
\sum _{x\in\pscript _n}p_c^n(x)=1
\end{equation*}

\begin{exam}{1}  
Figure~2 illustrates the first two steps of a CSGP where the 2 indicates the multiplicity $m(x_3\to x_6)=2$. Table~1 lists the probabilities of the various sites for the general coupling constants $c_k$ and the particular coupling constants $c'_k=1/k!$ where $d=(c_0+c_1)(c_0+2c_1+c_2)$.
\end{exam}

\setlength{\unitlength}{8pt}
\begin{picture}(45,26)
\put(22,8){\circle{7}}  
\put(22,8){\circle*{.35}}
\put(17.5,7){\makebox{$x_1$}}
\put(19.5,8){\vector(-1,1){4}} 
\put(24.5,8){\vector(1,1){4}} 
\put(14,14){\circle{7}} 
\put(12,15.5){\vector(-1,1){5.5}} 
\put(16,15,5){\vector(1,1){5}} 
\put(14,15){\circle*{.35}}
\put(14,13){\circle*{.35}}
\put(14,13.2){\line(0,1){1.5}}
\put(9.75,13.75){\makebox{$x_2$}}
\put(14,16.5){\vector(0,1){4}} 
\put(30,14){\circle{7}} 
\put(28,15.5){\vector(-1,1){5}} 
\put(32,15,5){\vector(1,1){5}} 
\put(24,17){\makebox{$2$}}
\put(29,14){\circle*{.35}}
\put(31,14){\circle*{.35}}
\put(33,13.5){\makebox{$x_3$}}
\put(30,16.5){\vector(0,1){4}} 
\put(5,23){\circle{7}} 
\put(5,21.5){\circle*{.35}}
\put(5,23){\circle*{.35}}
\put(5,24.5){\circle*{.35}}
\put(5,21.5){\line(0,1){3}}
\put(.75,23){\makebox{$x_4$}}
\put(14,23){\circle{7}} 
\put(14,22){\circle*{.35}}
\put(13.25,23.5){\circle*{.35}}
\put(14.75,23.5){\circle*{.35}}
\put(14,22){\line(1,2){.8}}
\put(14,22){\line(-1,2){.8}}
\put(10,23){\makebox{$x_5$}}
\put(22,23){\circle{7}} 
\put(21,22){\circle*{.35}}
\put(22.5,22){\circle*{.35}}
\put(21,23.5){\circle*{.35}}
\put(21,22){\line(0,1){1.5}}
\put(18,23){\makebox{$x_6$}}
\put(30,23){\circle{7}} 
\put(29,22){\circle*{.35}}
\put(31,22){\circle*{.35}}
\put(30,23.75){\circle*{.35}}
\put(29,22){\line(1,2){.8}}
\put(31,22){\line(-1,2){.8}}
\put(26,23){\makebox{$x_7$}}
\put(38,23){\circle{7}} 
\put(36.5,23){\circle*{.35}}
\put(38,23){\circle*{.35}}
\put(39.5,23){\circle*{.35}}
\put(41,23){\makebox{$x_8$}}
\end{picture}
\vskip -2pc
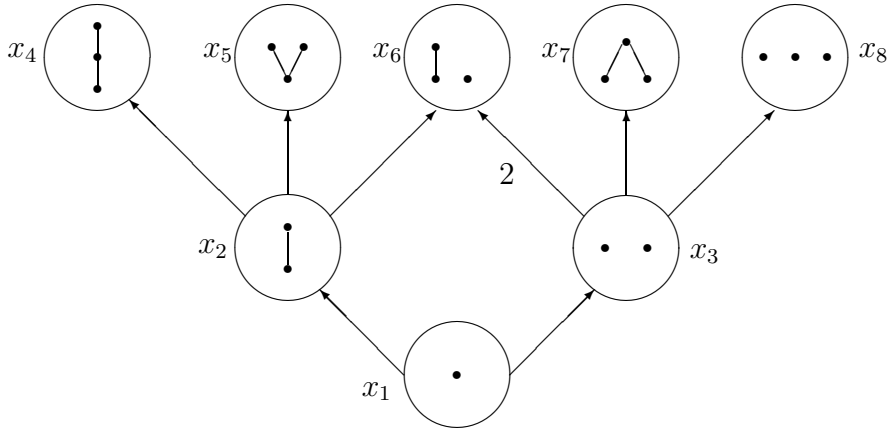
\centerline{\textbf{Figure 2}}

\begin{center}  
\begin{tabular}{c|c|c|c|c|c|c|c|c}
$x_i$&$x_1$&$x_2$&$x_3$&$x_4$&$x_5$&$x_6$&$x_7$&$x_8$\\
\hline
$p_c^{(n)}(x_i)$&$1$&$\frac{c_1}{c_0+c_1}$&$\frac{c_0}{c_0+c_1}$&$\frac{c_1(c_1+c_2)}{d}$%
&$\frac{c_1^2}{d}$&$\frac{3c_0c_1}{d}$&$\frac{c_0c_2}{d}$&$\frac{c_0^2}{d}$\\
\hline
$p_{c'}^n(x_i)$&$1$&$\frac{1}{2}$&$\frac{1}{2}$&$\frac{3}{14}$&$\frac{1}{7}$%
&$\frac{3}{7}$&$\frac{1}{14}$&$\frac{1}{7}$\\
\noalign{\bigskip}
\multicolumn{9}{c}%
{\textbf{Table 1}}\\
\end{tabular}
\end{center}
\vskip 3pc

For $A\subseteq\Omega _n$ we use the notation
\begin{equation*}
A\Rightarrow =\cup\brac{\omega\Rightarrow\colon\omega\in A}
\end{equation*}
Thus, $A\Rightarrow$ is the set of paths whose initial $n$-paths are elements of $A$. We call $A\Rightarrow$ a
\textit{cylinder set} and define
\begin{equation*}
\ascript _n=\brac{A\Rightarrow\colon A\subseteq\Omega _n}
\end{equation*}
In particular, if $\omega\in\Omega _n$ then the \textit{elementary cylinder set} $\rmcyl (\omega )$ is given by
$\rmcyl (\omega )=\omega\Rightarrow$. It is easy to check that the $\ascript _n$ form an increasing sequence
$\ascript _1\subseteq\ascript _2\subseteq\cdots$ of algebras on $\Omega$ and hence
$\cscript (\Omega )=\cup\ascript _n$ is an algebra of subsets of $\Omega$. Also for $A\in\cscript (\Omega )$ of the form $A=A_1\Rightarrow$, $A_1\subseteq\Omega _n$,  we define $p_c(A)=p_c^n(A_1)$. It is easy to check that
$p_c$ is a well-defined probability measure on $\cscript (\Omega )$. It follows from the Kolmogorov extension theorem that
$p_c$ has a unique extension to a probability measure $\nu _c$ on the $\sigma$-algebra $\ascript$ generated by
$\cscript (\Omega )$. We conclude that $(\Omega ,\ascript ,\nu _c)$ is a probability space, the increasing sequence of subalgebras $\ascript _n$ generates $\ascript$ and that the restriction $\nu _c\mid\ascript _n=p_c^n$. Hence, the subspaces $H_n=L_2(\Omega ,\ascript _n,p_c^n)$ form a filtration of the Hilbert space
$H=L_2(\Omega ,\ascript ,\nu _c)$.

\section{Quantum Sequential Growth Processes} 
This section employs the framework of Section~2 to obtain a quantum sequential growth process (QSGP) from the CSGP $(\pscript ,p_c)$ developed in Section~3. We have seen that the $n$-\textit{path Hilbert space}
$H_n=L_2(\Omega ,\ascript _n,p_c^n)$ forms a filtration of the \textit{path Hilbert space}
$H=L_2(\Omega ,\ascript ,\nu _c)$. In the sequel, we assume that $p_c^n(\omega )\ne 0$ for every
$\omega\in\Omega _n$, $n=1,2,\ldots\,$. Then the set of vectors
\begin{equation*}
e_\omega ^n=p_c^n(\omega )^{1/2}\chi _{\rmcyl (\omega )},\omega\in\Omega _n
\end{equation*}
form an orthonormal basis for $H_n$, $n=1,2,\ldots\,$. For $A\in\ascript _n$, notice that $\chi _A\in H$ with
$\doubleab{\chi _A}=p_c^n(A)^{1/2}$.

We call a DQP $\rho _n\in\!\qscript (H_n)$ a \textit{quantum sequential growth process} (QSGP). We call $\rho _n$ the
\textit{local operators} and $\mu _n(A)=D_n(A,A)$ the \textit{local} $q$-\textit{measures} for the process. If
$\rho =\lim\rho _n$ exists in the strong operator topology, then $\rho$ is a $q$-probability operator on $H$ called the
\textit{global operator} for the process. If the global operator $\rho$ exists, then
$\muhat (A)=\elbows{\rho\chi _A,\chi _A}$ is a (continuous) $q$-measure on $\ascript$ that extends $\mu _n$,
$n=1,2,\ldots\,$. Unfortunately, the global operator does not exist, in general, so we must be content to work with the local operators \cite{dgt08, gud11, sor11}. In this case, we still have the $q$-measure $\mu$ on the quadratic algebra
$\sscript (\rho _n)\subseteq\ascript$ that extends $\mu _n$ $n=1,2,\ldots\,$. We frequently identify a set
$A\subseteq\Omega _n$ with the corresponding cylinder set $(A\Rightarrow )\in\ascript _n$. We then have the
$q$-measure, also denoted by $\mu _n$, on $2^{\Omega _n}$ defined by $\mu _n(A)=\mu _n(A\Rightarrow )$. Moreover, we define the $q$-measure, again denoted by $\mu _n$, on $\pscript _n$ by
\begin{equation*}
\mu _n(A)=\mu _n\paren{\brac{\omega\in\Omega _n\colon\omega\hbox{ end in }A}}
\end{equation*}
for all $A\subseteq\pscript _n$. In particular, for $x\in\pscript _n$ we have
\begin{equation*}
\mu _n\paren{\brac{x}}=\mu _n\paren{\brac{\omega\in\Omega _n\colon\omega\hbox{ ends with }x}}
\end{equation*}

If $A\in\ascript _n$ has the form $A_1\Rightarrow$ for $A_1\subseteq\Omega _n$ then $A\in\ascript _{n+1}$ and 
$A=(A_1\to )\Rightarrow$ where $A_1\to\subseteq\Omega _{n+1}$. Let $\rho _n\in\qscript (H_n)$,
$\rho _{n+1}\in\qscript (H_{n+1})$ and let $D_n(A,B)=\elbows{\rho _n\chi _B,\chi _A}$,
$D_{n+1}(A,B)=\elbows{\rho _{n+1}\chi _B,\chi _A}$ be the corresponding decoherence functionals. Then
$\rho _{n+1}$ is consistent with $\rho _n$ if and only if for all $A,B\subseteq\Omega _n$ we have
\begin{equation}         
\label{eq41}
D_{n+1}\sqbrac{(A\to )\Rightarrow,(B\to )\Rightarrow}=D_n(A\Rightarrow, B\Rightarrow )
\end{equation}

\begin{lem}       
\label{lem41}
For $\rho _n\in\qscript (H_n)$, $\rho _{n+1}\in\qscript (H_{n+1})$ we have that $\rho _{n+1}$ is consistent with
$\rho _n$ if and only if for all $\omega ,\omega '\in\Omega _n$ we have
\begin{equation}         
\label{eq42}
D_{n+1}\sqbrac{(\omega\to )\Rightarrow ,(\omega '\to )\Rightarrow}=D_n(\omega\Rightarrow, \omega '\Rightarrow )
\end{equation}
\end{lem}
\begin{proof}
Necessity is clear. For sufficiency, suppose \eqref{eq42} holds. Then for every $A,B\subseteq\Omega _n$ we have
\begin{align*}
D_{n+1}\sqbrac{(A\to )\Rightarrow ,(B\to )\Rightarrow )}
  &=\sum _{\omega\in A}\sum _{\omega '\in B}D_{n+1}
  D_{n+1}\sqbrac{(\omega\to )\Rightarrow ,(\omega '\to)\Rightarrow}\\
  &=\sum _{\omega\in A}\sum _{\omega '\in B}D_n(\omega\Rightarrow ,\omega '\Rightarrow )
  =D_n(A\Rightarrow ,B\Rightarrow )
\end{align*}
and the result follows from \eqref{eq41}.
\end{proof}

For $\omega =x_1x_2\cdots x_n\in\Omega _n$ and $y\in\pscript _{n+1}$ with $x_n\to y$ we use the notation
$\omega y\in\Omega _{n+1}$ where $\omega y=x_1x_2\cdots x_ny$. We also define
$p_c (\omega\to y)=p_c(x_n\to y)$ and write $\omega\to y$ whenever $x_n\to y$.

\begin{thm}       
\label{thm42}
For $\rho _n\in\qscript (H_n)$, $\rho _{n+1}\in\qscript (H_{n+1})$ we have that $\rho _{n+1}$ is consistent with 
$\rho _n$ if and only if for every $\omega ,\omega '\in\Omega _n$ we have
\begin{equation}         
\label{eq43}
\elbows{\rho _ne _{\omega '}^n,e_\omega ^n}
  =\sum _{\substack{x\in\pscript _{n+1}\\\omega '\to x}}\sum _{\substack{y\in\pscript _{n+1}\\\omega\to y}}
  p_c(\omega '\to x)^{1/2}p_c(\omega\to y)^{1/2}\elbows{\rho _{n+1}e_{\omega 'x}^{n+1}e_{\omega y}^{n+1}}
\end{equation}
\end{thm}
\begin{proof}
By Lemma~\ref{lem41}, $\rho _{n+1}$ is consistent with $\rho _n$ if and only if \eqref{eq42} holds. But
\begin{align*}
D_n(\omega\Rightarrow ,\omega '\Rightarrow )
  &=\elbows{\rho _n\chi _{\omega '\Rightarrow},\chi _{\omega\Rightarrow}}
  =\elbows{\rho _n\chi _{\rmcyl (\omega ')},\chi _{\rmcyl (\omega )}}\\
  &=p_c^n(\omega ')^{1/2}p_c^n(\omega )^{1/2}\elbows{\rho _ne_{\omega '}^n,e_\omega ^n}
\end{align*}
Moreover, we have
\begin{align*}
D_{n+1}\sqbrac{(\omega\to )\Rightarrow ,(\omega '\to)\Rightarrow}
  &=\elbows{\rho _{n+1}\chi _{(\omega '\to)\Rightarrow},\chi _{(\omega\to )\Rightarrow}}\\
  &=\sum _{\substack{x\in\pscript _{n+1}\\\omega '\to x}}\sum _{\substack{y\in\pscript _{n+1}\\\omega\to y}}
  \elbows{\rho _{n+1}\chi _{\omega 'x\Rightarrow},\chi _{\omega y\Rightarrow}}\\
  &=\sum _{\substack{x\in\pscript _{n+1}\\\omega '\to x}}\sum _{\substack{y\in\pscript _{n+1}\\\omega\to y}}
  \elbows{\rho _{n+1}\chi _{\rmcyl (\omega 'x)},\chi _{\rmcyl (\omega y)}}\\
  &\hskip -10pc
  =\sum _{\substack{x\in\pscript _{n+1}\\\omega '\to x}}\sum _{\substack{y\in\pscript _{n+1}\\\omega\to y}}
  p_c^n(\omega 'x)^{1/2}p_c^n(\omega y)^{1/2}\elbows{\rho _{n+1}e_{\omega 'x}^{n+1},e_{\omega y}^{n+1}}\\
  &\hskip -10pc
  =p_c^n(\omega ')^{1/2}p_c^n(\omega )^{1/2}
  \sum _{\substack{x\in\pscript _{n+1}\\\omega '\to x}}\sum _{\substack{y\in\pscript _{n+1}\\\omega\to y}}
  p_c(\omega '\to x)p_c(\omega\to y)^{1/2}\elbows{\rho _ne_{\omega 'x}^{n+1},e_{\omega y}^{n+1}}
\end{align*}
The result now follows.
\end{proof}

Viewing $H_n$ as $L_2(\Omega _n,2^{\Omega _n},p_c^n)$ we can write \eqref{eq43} in the simple form
\begin{equation}         
\label{eq44}
\elbows{\rho _n\chi _{\brac{\omega '}},\chi _{\brac{\omega}}}
  =\elbows{\rho _{n+1}\chi _{\omega '\to},\chi _{\omega\to}}
\end{equation}

\begin{cor}       
\label{cor43}
A sequence $\rho _n\in\qscript (H_n)$ is a QSGP if and only if \eqref{eq43} or \eqref{eq44} hold for every
$\omega ,\omega '\in\Omega _n$, $n=1,2,\ldots\,$.
\end{cor}

We now consider pure $q$-probability operators. In the following results we again view $H_n$ as
$L_2(\Omega _n,2^{\Omega _n},p_c^n)$.

\begin{cor}       
\label{cor44}
If $\rho _n\in\qscript _p(H_n)$, $\rho _{n+1}\in\qscript _p(H_{n+1})$ with $p_n=\ket{\psi _n}\bra{\psi _n}$,
$\rho _{n+1}=\ket{\psi _{n+1}}\bra{\psi _{n+1}}$, then $\rho _{n+1}$ is consistent with $\rho _n$ if and only if for
every $\omega ,\omega '\in\Omega _n$ we have
\begin{equation}         
\label{eq45}
\elbows{\psi _n,\chi _{\brac{\omega}}}\elbows{\chi _{\brac{\omega '}},\psi _n}
=\elbows{\psi _{n+1},\chi _{\omega\to}}\elbows{\chi _{\omega '\to},\psi _{n+1}}
\end{equation}
\end{cor}

\begin{cor}       
\label{cor45}
A sequence $\ket{\psi _n}\bra{\psi _n}\in\qscript _p(H_n)$ is a QSGP if and only if \eqref{eq45} holds for every
$\omega ,\omega '\in\Omega _n$.
\end{cor}

We say that $\psi _{n+1}\in\vscript (H_{n+1})$ is \textit{strongly consistent} with $\psi _n\in\vscript (H_n)$ if for every
$\omega\in\Omega _n$ we have
\begin{equation}         
\label{eq46}
\elbows{\psi _n,\chi _{\brac{\omega}}}=\elbows{\psi _{n+1},\chi _{\omega\to}}
\end{equation}
By \eqref{eq45} strong consistency implies the consistency of the corresponding $q$-probability operators.

\begin{cor}       
\label{cor46}
If $\psi _{n+1}\in\vscript (H_{n+1})$ is strongly consistent with $\psi _n\in\vscript (H_n)$, $n=1,2,\ldots$, then
$\ket{\psi _n}\bra{\psi _n}\in\qscript _p(H_n)$ is a QSGP.
\end{cor}

\begin{lem}       
\label{lem47}
If $\psi _n\in\vscript (H_n)$ and $\psi _{n+1}\in H_{n+1}$ satisfies \eqref{eq46} for every $\omega\in\Omega _n$, then
$\psi _{n+1}\in\vscript (H_{n+1})$.
\end{lem}
\begin{proof}
Since $\psi _n\in\vscript (H_n)$ we have by \eqref{eq46} that
\begin{align*}
\ab{\elbows{\psi _{n+1},1}}&=\ab{\sum _{\omega\in\Omega _n}\elbows{\psi _{n+1},\chi _{\omega\to}}}
  =\ab{\sum _{\omega\in\Omega _n}\elbows{\psi _n,\chi _{\brac{\omega}}}}\\
  &=\ab{\elbows{\psi _n,1}}=1
\end{align*}
The result now follows.
\end{proof}

\begin{cor}       
\label{cor48}
If $\doubleab{\psi _1}=1$ and $\psi _n\in H_n$ satisfies \eqref{eq46} for all $\omega\in\Omega _n$, $n=1,2,\ldots$,
then $\ket{\psi _n}\bra{\psi _n}$ is a QSGP.
\end{cor}
\begin{proof}
Since $\doubleab{\psi _1}=1$, it follows that $\psi _1\in\vscript (H_1)$. By Lemma~\ref{lem47},
$\psi _n\in\vscript (H_n)$, $n=1,2,\ldots\,$. Since \eqref{eq46} holds, the result follows from Corollary~\ref{cor46}.
\end{proof}

Another way of writing \eqref{eq46} is
\begin{equation}         
\label{eq47}
\sum _{\omega\to x}p_c^{n+1}(\omega x)\psi _{n+1}(\omega x)=p_c^n(\omega )\psi _n(x)
\end{equation}
for every $\omega\in\Omega _n$.

\section{Discrete Quantum Gravity Models} 
This section gives some examples of QSGP that can serve as models for discrete quantum gravity. The simplest way to construct a QSGP is to form the constant pure DQP $\rho _n=\ket{1}\bra{1}$, $n=1,2,\ldots\,$. To show that
$\rho _n$ is indeed consistent, we have for $\omega\in\Omega _n$ that
\begin{equation*}
\sum _{\omega\to x}p_c^{n+1}(\omega x)=\sum _{\omega\to x}p_c^n(\omega )p_c(\omega\to x)
  =p_c^n(\omega )\sum _{\omega\to x}p_c(\omega\to x)=p_c^n(\omega )
\end{equation*}
so consistency follows from \eqref{eq47}. The corresponding $q$-measures are given by
\begin{equation*}
\mu _n(A)=\ab{\elbows{1,\chi _A}}^2=p_c^n(A)^2
\end{equation*}
for every $A\in\ascript _n$. Hence, $\mu _n$ is the square of the classical measure. Of course, $\ket{1}\bra{1}$ is the global $q$-probability operator for this QSGP and in this case $\sscript (\rho _n)=\ascript$. Moreover, we have the global $q$-measure $\mu (A)=\nu _c(A)^2$ for $A\in\ascript$.

Another simple way to construct a QSGP is to employ Corollary~\ref{cor48}. In this way we can let $\psi _1=1$,
$\psi _2$ any vector in $L_2(\Omega _2,2^{\Omega _2},p_c^2)$ satisfying
\begin{equation*}
\elbows{\psi _2,\chi _{\brac{x_1x_2}}}+\elbows{\psi _2\chi _{\brac{x_1x_3}}}=\elbows{\psi _1,\chi _{\brac{x_1}}}=1
\end{equation*}
and so on, where $x_1,x_2,x_3$ are given in Figure~2. As a concrete example, let $\psi _1=1$,
\begin{equation*}
\psi _2=\tfrac{1}{2}\sqbrac{p_c^2(x_1x_2)^{-1}\chi _{\brac{x_1x_2}}+p_c^2(x_1x_3)\chi _{\brac{x_1x_3}}}
\end{equation*}
and in general
\begin{equation*}
\psi _n=\frac{1}{\ab{\Omega _n}}\sum _{\omega\in\Omega _n}p_c^n(\omega )^{-1}\chi _{\brac{\omega}}
\end{equation*}
The $q$-measure $\mu _1$ is $\mu _1\paren{\brac{x_1}}=1$ and $\mu _2$ is given by
\begin{align*}
\mu _2\paren{\brac{x_1x_2}}&=\ab{\elbows{\psi _2,\chi _{\brac{x_1x_2}}}}^2=\tfrac{1}{4}\\
\mu _2\paren{\brac{x_1x_3}}&=\ab{\elbows{\psi _2,\chi _{\brac{x_1x_3}}}}^2=\tfrac{1}{4}\\
\mu _2(\Omega _2)&=\ab{\elbows{\psi _2,1}}^2=1
\end{align*}
In general, we have $\mu _n(A)=\ab{A}^2/\ab{\Omega _n}^2$ for all $A\in\Omega _n$. Thus $\mu _n$ is the
square of the uniform distribution. The global operator does not exist because there is no $q$-measure on $\ascript$ that extends $\mu _n$ for all $n\in\positive$. For $A\in\ascript$ we have
\begin{equation*}
\elbows{\psi _n,\chi _A}=\int\psi _n\chi _Ad\nu _c
  =\frac{\ab{A\cap\brac{\rmcyl (\omega )\colon\omega\in\Omega _n}}}{\ab{\Omega _n}}
\end{equation*}
Letting $\rho _n=\ket{\psi _n}\bra{\psi _n}$ we conclude that $A\in\sscript (\rho _n)$ if and only if
\begin{equation*}
\lim _{n\to\infty}\frac{\ab{A\cap\brac{\rmcyl (\omega )\colon\omega\in\Omega _n}}}{\ab{\Omega _n}}
\end{equation*}
exists. For example, if $\ab{A}<\infty$ then for $n$ sufficiently large we have
\begin{equation*}
\ab{A\cap\brac{\rmcyl (\omega )\colon\omega\in\Omega _n}}=\ab{A}
\end{equation*}
so $A\in\sscript (\rho _n)$ and $\mu (A)=0$. In a similar way if $\ab{A}<\infty$ then for the complement $A'$, if $n$ is sufficiently large we have
\begin{equation*}
\ab{A'\cap\brac{\rmcyl (\omega )\colon\omega\in\Omega _n}}=\ab{\Omega _n}-\ab{A}
\end{equation*}
so $A'\in\sscript (\rho _n)$ with $\mu (A')=1$.

We now present another method for constructing a QSGP. Unlike the previous method this DQP is not pure. Let
$\alpha _\omega\in\complex$, $\omega\in\Omega _n$ satisfy
\begin{equation}         
\label{eq51}
\ab{\sum _{\omega\in\Omega _n}\alpha _\omega p_c^n(\omega )^{1/2}}=1
\end{equation}
and let $\rho _n$ be the operator on $H_n$ satisfying
\begin{equation}         
\label{eq52}
\elbows{\rho _ne_\omega ^n,e_{\omega '}^n}=\alpha _{\omega '}\overline{\alpha _\omega}
\end{equation}
Then $\rho _n$ is a positive operator and by \eqref{eq51}, \eqref{eq52} we have
\begin{align*}
\elbows{\rho _n1,1}
&=\elbows{\rho _n\sum _\omega p_c^n(\omega )^{1/2}e_\omega ^n,
  \sum _{\omega '}p_c^n(\omega ')^{1/2}e_{\omega '}^n}\\
  &=\sum _{\omega ,\omega '}p_c^n(\omega )^{1/2}p_c^n(\omega ')^{1/2}
  \elbows{\rho _ne_\omega ^n,e_{\omega '}^n}\\
  &=\ab{\sum _\omega p_c^n(\omega )^{1/2}\alpha _\omega}^2=1
\end{align*}
Hence, $\rho _n\in\qscript (H_n)$. Now
\begin{equation*}
\Omega _{n+1}=\brac{\omega x\colon\omega\in\Omega _n,x\in\pscript _{n+1},\omega\to x}
\end{equation*}
and for each $\omega x\in\Omega _{n+1}$, let $\beta _{\omega x}\in\complex$ satisfy
\begin{equation*}
\ab{\sum _{\omega x\in\Omega _{n+1}}\beta _{\omega x}p_c^{n+1}(\omega x)^{1/2}}=1
\end{equation*}
Let $\rho _{n+1}$ be the operator on $H_{n+1}$ satisfying
\begin{equation}         
\label{eq53}
\elbows{\rho _{n+1}e_{\omega x}^{n+1},e_{\omega 'x'}^{n+1}}=\beta _{\omega 'x'}\overline{\beta _{\omega x}}
\end{equation}
As before, we have that $\rho _{n+1}\in\qscript (H_{n+1})$. The next result follows from Theorem~\ref{thm42}.

\begin{thm}       
\label{thm51}
The operator $\rho _{n+1}$ is consistent with $\rho _n$ if and only if for every $\omega ,\omega'\in\Omega _n$ we have
\begin{equation}         
\label{eq54}
\alpha _{\omega '}\overline{\alpha _\omega}
 =\sum _{\substack{x'\in\pscript _{n+1}\\\omega '\to x'}}\beta _{\omega 'x'}p_c(\omega '\to x')^{1/2}
 \sum _{\substack{x\in\pscript _{n+1}\\\omega\to x}}\overline{\beta _{\omega x}}p_c(\omega\to x)^{1/2}
\end{equation}
\end{thm}

A sufficient condition for \eqref{eq54} to hold is
\begin{equation}         
\label{eq55}
\sum _{\substack{x\in\pscript _{n+1}\\\omega\to x}}\beta _{\omega x}p_c(\omega\to x)^{1/2}=\alpha _\omega
\end{equation}
The proof of the next result is similar to the proof of Lemma~\ref{lem47}.

\begin{lem}       
\label{lem52}
Let $\rho _n\in\qscript (H_n)$ be defined by \eqref{eq52} and let $\rho _{n+1}$ be the operator on $H_{n+1}$ defined by \eqref{eq53}. If \eqref{eq55} holds, then $\rho _{n+1}\in\qscript (H_{n+1})$ and $\rho _{n+1}$ is consistent with
$\rho _n$.
\end{lem}

The next result gives the general construction.

\begin{cor}       
\label{cor53}
Let $\rho _1=I\in\qscript (H_1)$ and define $\rho _n\in\qscript (H_n)$ inductively by \eqref{eq53}. Then
$\rho _n$ is a QSGP.
\end{cor}

\end{document}